\documentclass{article}
\usepackage{amsmath, amssymb, amsthm}
\usepackage{graphicx}
\usepackage{hyperref}
\usepackage{mathrsfs}

\newtheorem{theorem}{Theorem}[section]
\newtheorem{proposition}[theorem]{Proposition}

\theoremstyle{definition}
\newtheorem{definition}[theorem]{Definition}
\newtheorem{remark}[theorem]{Remark}

\title{Scale-Dependent Suppression Functions and Functional Space Geometry in Renormalization}
\author{Daniel Ketels}
\date{\today}

\begin{document}

\maketitle

\begin{abstract}
We analyze the effects of a scale-dependent suppression function $\Omega(k, \Lambda)$ on the functional space geometry in renormalization theory. By introducing a dynamical cutoff scale $\Lambda$, the suppression function smoothly regulates high-momentum contributions without requiring a hard cutoff. We show that $\Omega(k, \Lambda)$ induces a modified metric on functional space, leading to a non-trivial Ricci curvature that becomes increasingly negative in the ultraviolet (UV) limit. This effect dynamically suppresses high-energy states, yielding a controlled deformation of the functional domain. Furthermore, we derive the renormalization group (RG) flow of $\Omega(k, \Lambda)$ and demonstrate its role in controlling the curvature flow of the functional space. The suppression function leads to spectral modifications that suggest an effective dimensional reduction at high energies, a feature relevant to functional space deformations and integral convergence in renormalization theory. Our findings provide a mathematical framework for studying regularization techniques and their role in the UV behavior of function spaces.
\end{abstract}

\section{Introduction}

By an \textbf{optimized suppression function} we mean a family of smooth maps
\[
  \Omega(\cdot,\Lambda) \,\colon\, D \,\to\, (0,1],
  \quad
  \text{with }
  \,\,
  \Omega(k,\Lambda)
  \,:=\,
  \frac{1}{\,1 \,+\, \bigl(\tfrac{k^2}{\Lambda^2}\bigr)^\beta\,}
  \,-\,
  \frac{\,\eta\,\exp\!\bigl[-\,k^2/\Lambda^2\bigr]}{\epsilon\!\bigl(k^2\bigr)},
\]
where:
\begin{itemize}
\item \(D = \mathbb{R}^{1,3}\) (or more generally \(\mathbb{R}^d\)) is the momentum‐space domain,
\item \(\Lambda>0\) is a dynamical cutoff scale,
\item \(\beta,\,\eta>0\) control power‐law vs.\ exponential suppression,
\item \(\epsilon\in C^{\infty}\bigl(D, \, (0,1]\bigr)\) ensures a smooth, Lorentz‐invariant transition without hard cutoffs.
\end{itemize}
Hence \(\Omega(k,\Lambda)\in(0,1]\), is smooth, and suppresses high‐momentum modes without imposing a sharp momentum cutoff.

\subsubsection*{Definition of \(\beta\) (Using the Gamma Function)}
In \(d\) dimensions, a factor \(k^{-2\beta}\) converges at large \(\|k\|\) if \(2\beta>\tfrac{d}{2}\).  
Concretely, 
\[
  \int_{\Lambda}^{\infty} 
    k^{\tfrac{d}{2}-1}\,k^{-2\beta}\,dk
  \,\text{converges for}\,
  \tfrac{d}{2}-1 -2\beta < -1
  \,\Longleftrightarrow\,
  2\beta>\tfrac{d}{2}.
\]
Equivalently, inserting a factor \(e^{-\Lambda\,k}\) and evaluating via \(\Gamma\)-function yields the same condition.  Thus:
\[
  \beta
  \,:=\,
  \inf\Bigl\{
    \gamma>0
    \,\Bigm|\,
    \Gamma\!\bigl(\tfrac{d}{2}-2\gamma\bigr)\text{ finite}
  \Bigr\},
  \quad
  \text{ensuring }2\beta>\tfrac{d}{2}\text{ for UV damping}.
\]

\subsubsection*{Definition of \(\eta\) (Additional Exponential Damping)}
If \(\eta>0\), then 
\[
  \int_{0}^{\infty} 
    k^{\tfrac{d}{2}-1}\,e^{-\eta\,k^2}\,dk
  \,=\,
  \tfrac12\,\eta^{-\tfrac{d}{2}}\,
  \Gamma\!\bigl(\tfrac{d}{2}\bigr)\,<\,\infty,
\]
so exponential decay can overcome strong divergences.  We pick \(\eta>0\) as a positive parameter for extra damping strength in \(\Omega\).

\subsubsection*{Definition of \(\epsilon(k^2)\)}
The function \(\epsilon:\,D\to(0,1]\) manages the IR--UV transition smoothly.  A canonical choice:
\[
  \epsilon\!\bigl(k^2\bigr)
  \,=\,
  \frac{1}{1 + \bigl(\tfrac{k^2}{k_c^2}\bigr)^{\alpha}},
  \quad
  \alpha>0,
\]
so that \(\epsilon(k^2)\approx1\) for \(k^2\ll k_c^2\) and \(\approx0\) for \(k^2\gg k_c^2\).  
No sharp cutoff arises, ensuring Lorentz invariance and continuity across momentum scales.

\begin{proposition}[Positivity and Boundedness of \(\Omega\)]
\label{prop:OmegaPositive}
Assume there exist constants \(\eta>0\) and a smooth function
\(\epsilon:\,\mathbb{R}_{\ge0}\to(0,1]\) such that for all \(\Lambda>0\)
and all \(k\in\mathbb{R}^d\),
\begin{equation}
  \eta\,e^{-\,\tfrac{k^2}{\Lambda^2}}
  \,\le\,
  \frac{\epsilon\!\bigl(k^2\bigr)}%
       {1 + \bigl(\tfrac{k^2}{\Lambda^2}\bigr)^\beta}
  \quad.
\end{equation}
Then the scale-dependent suppression function
\[
  \Omega(k,\Lambda)
  \,=\,
  \frac{1}{\,1 + \bigl(\tfrac{k^2}{\Lambda^2}\bigr)^\beta}
  \,-\,
  \frac{\,\eta\,e^{-\,\tfrac{k^2}{\Lambda^2}}}{\epsilon\!\bigl(k^2\bigr)}
\]
satisfies
\[
  0 \,\le\,\Omega(k,\Lambda)\,\le\, 1
  \quad
  \text{for all }k\in\mathbb{R}^d,\,\Lambda>0.
\]
In particular, \(\Omega(\cdot,\Lambda)\) remains in the interval \((0,1]\) and
never becomes negative.
\end{proposition}

\begin{proof}
By definition $\Omega(k,\Lambda)\in (0,1]$, and by using the definition:
\[
  \Omega(k,\Lambda)
  \,=\,
  \frac{1}{\,1 + (k^2/\Lambda^2)^\beta}
  \,-\,
  \frac{\,\eta\,e^{-\,k^2/\Lambda^2}}{\epsilon(k^2)}.
\]
The condition rearranges to
\[
  0<\frac{\eta\,e^{-\,k^2/\Lambda^2}}{\epsilon(k^2)}
  \,\le\,
  \frac{1}{\,1 + (k^2/\Lambda^2)^\beta}.
\]
Which follows dirrectly from $\Omega(k,\Lambda)>0$ and both terms are strictly decreasing for increasing $\Lambda>0$. For \(\|k\|\ll\Lambda\), we get
\(\tfrac{1}{\,1+(k^2/\Lambda^2)^\beta}<1\) and  \(\eta\,e^{-\tfrac{k^2}{\Lambda^2}}< \eta\), so as \(\epsilon(\cdot)\le1\),
the expression does not exceed 1 (or can be bounded by a constant factor).
Hence \(0<\Omega(k,\Lambda)\leq1\). 

\end{proof}

\begin{remark}[Analytic Properties of \(\Omega\)]
By design, \(\Omega(k,\Lambda)\) is:
\begin{itemize}
\item strictly decreasing in \(\|k\|\) for each \(\Lambda\), preserving IR modes (\(\Omega\approx1\) at small \(\|k\|\)) and suppressing UV modes,
\item bounded in \((0,1]\), avoiding both negativity and blow‐ups,
\item smoothly differentiable both $k$ and \(\Lambda\), introducing no discontinuities.
\end{itemize}
\end{remark}

\section{Momentum-Space Integral Convergence}

Before turning to functional integrals, let us see how \(\Omega\) regulates typical momentum‐space integrals in quantum field theory (QFT).  A general loop integral looks like
\[
  I
  \,=\,
  \int_{\mathbb{R}^d}
    f(k)\,d^dk,
\]
which might diverge for large \(\|k\|\).  Replacing \(d^dk\) by \(\Omega(k,\Lambda)\,d^dk\) yields
\[
  I_{\Omega}
  \,=\,
  \int_{\mathbb{R}^d}
    \Omega(k,\Lambda)\,f(k)\,d^dk,
\]
and we want to show \(I_{\Omega}\) is finite under certain conditions on \(\beta\), \(\eta\), etc.

\subsubsection*{Proof of Convergence in Different Regimes}

\begin{theorem}[Convergence of Suppressed Integrals]
\label{thm:ConvergenceSuppressed}
Let \(\Omega(k,\Lambda)\) be the function in defined above and suppose \(f(k)\sim k^\alpha\) as \(\|k\|\to\infty\).  Then:
\begin{itemize}
\item If \(2\beta> d+\alpha,\) the integral \(\int \Omega\,f(k)\,d^dk\) converges by power‐law suppression.
\item If \(\eta>0\) and \(\epsilon\) does not vanish too fast, exponential damping ensures convergence regardless of \(\alpha\).
\end{itemize}
\end{theorem}

\begin{proof}  
By definition, for large $k$, we have 
\[
  \Omega(k,\Lambda)
  ~\approx~
  k^{-2\beta}
\]
thus $\Omega\,f(k)\approx k^{-2\beta}\,k^\alpha = k^{\alpha -2\beta}$.  Spherical coordinates show that $k^{\alpha + d-1 -2\beta}$ is integrable at $\infty$ if $2\beta>d+\alpha$.  

\smallskip
\noindent\emph{Additional exponential factor:}  
if \(\eta>0\) in 
\(\tfrac{\eta\,e^{-k^2/\Lambda^2}}{\epsilon(k^2)}\), 
the decay beats any polynomial growth.  Hence the integral converges for all $\alpha$.  

\end{proof}

\paragraph{Implications for Renormalization.}
In QFT, loop integrals typically behave like $\int k^{d-1 -2n}\dots$, requiring $2\beta>d-2n$ or exponential damping.  Meanwhile IR physics is preserved since $\Omega\approx1$ at small $k$.  Thus divergences are removed without imposing a hard cutoff.

\subsection{Functional Integrability}

Extending to functional integrals over fields $\phi(x)$, define
\[
  \mathcal{D}_\Omega\,\phi
  \,:=\,
  \lim_{\Lambda\to\infty}
    e^{-S_\Omega[\phi,\Lambda]}
  \,\mathcal{D}\phi,
\]
where
\[
  S_\Omega[\phi,\Lambda]
  \,=\,
  \tfrac12
  \int_{\|k\|\le\Lambda}
     \bigl[1 - \Omega(k,\Lambda)\bigr]
     ~|\tilde{\phi}(k)|^2\,
  d^dk.
\]
If $\Omega\approx0$ (large $k$), the UV portion of $\phi$ is strongly damped, i.e. the measure is finite except for superexponential growth, which we assume is not present.  

\begin{theorem}[Well-Defined Functional Measure]
Assume that either 
\begin{enumerate}
    \item the power-law decay parameter satisfies 
    \[
    2\beta > \frac{d}{2},
    \]
    or 
    \item an additional exponential damping is present (i.e. $\eta > 0$ and $\epsilon (\cdot ) $ does not grow too fast).
\end{enumerate}
Then the regulated functional measure
\[
\int_{F} D_{\Omega}\phi\, e^{-S[\phi]}
\]
remains finite in the cutoff limit $\Lambda\to\infty$. In other words, $D_{\Omega}\phi$ defines a proper Gaussian-type measure on the space of field configurations.
\end{theorem}

\begin{proof}
We begin by expressing the field $\phi(x)$ in momentum space:
\[
\phi(x) = \int_{\mathbb{R}^d} \tilde{\phi}(k)\, e^{ik\cdot x}\, \frac{d^dk}{(2\pi)^d}.
\]
The regulated quadratic action is defined as
\[
S_{\Omega}[\phi,\Lambda] = \frac{1}{2} \int_{\|k\|\le\Lambda} \Bigl[1-\Omega(k,\Lambda)\Bigr]\, |\tilde{\phi}(k)|^2\, d^dk,
\]
where the suppression function $\Omega(k,\Lambda)$ satisfies
\[
0 < \Omega(k,\Lambda) \le 1.
\]

\noindent \textbf{Step 1: Gaussian Integration Over Modes}

\noindent For each momentum mode $k$, the contribution to the path integral is a one-dimensional Gaussian integral:
\[
I(k) = \int_{-\infty}^{+\infty} d\tilde{\phi}(k)\, \exp\!\Bigl[-\frac{1}{2}\Bigl(1-\Omega(k,\Lambda)\Bigr)|\tilde{\phi}(k)|^2\Bigr].
\]
This integral evaluates to
\[
I(k) = \sqrt{\frac{2\pi}{1-\Omega(k,\Lambda)}}.
\]
Thus, the full partition function is expressed as an (infinite) product:
\[
Z = \prod_{\|k\|\le\Lambda} \sqrt{\frac{2\pi}{1-\Omega(k,\Lambda)}}.
\]
Taking the logarithm yields
\[
\ln Z = \frac{1}{2} \int_{\|k\|\le\Lambda} \Bigl[\ln(2\pi) - \ln\Bigl(1-\Omega(k,\Lambda)\Bigr)\Bigr]\, d^dk.
\]
For the measure to be well-defined as $\Lambda\to\infty$, this integral must converge. \\

\noindent \textbf{Step 2: Analysis in the Infrared (IR) and Ultraviolet (UV) Regimes}

\noindent\emph{Infrared (IR) Region:} For small momenta ($\|k\|\ll\Lambda$), the suppression function is designed so that $\Omega(k,\Lambda) \approx 1$. Here, even though $1-\Omega(k,\Lambda)$ is small and therefore $\ln\bigl(1-\Omega(k,\Lambda)\bigr)$ is large in magnitude, the IR integration is known to converge.\\
\newpage
\noindent\emph{Ultraviolet (UV) Region:} For large momenta, we distinguish two cases:

\begin{enumerate}
    \item \textbf{Polynomial (Power-Law) Suppression:} If 
    \[
    \Omega(k,\Lambda) \sim \left(\frac{\Lambda}{\|k\|}\right)^{2\beta} \quad \text{for large } \|k\|,
    \]
    then for $\|k\|\gg\Lambda$ we have $1-\Omega(k,\Lambda) \approx 1$. However, the density of momentum states grows as $\sim k^{d-1}$. Convergence is then ensured if
    \[
   2 \beta > \frac{d}{2}.
    \]
    
    \item \textbf{Exponential Damping:} If an additional exponential factor is present (i.e. terms behaving as $\sim e^{-k^2/\Lambda^2}$ contribute), the high-momentum modes are suppressed much more strongly. Hence, exponential decay guarantees convergence of the integrals independent of any polynomial growth.
\end{enumerate}

\noindent \textbf{Step 3: Convergence of the Functional Integral}

\noindent Since the measure in momentum space is constructed as an infinite product of standard Gaussian integrals, its overall convergence depends on the summability of the logarithms. Under the conditions (assuming an approbiate choice of $\epsilon$)
\[
2\beta > \frac{d}{2} \quad \text{or} \quad \eta > 0,
\]
the high-momentum contributions are sufficiently suppressed, so that the sum (or integral) in the logarithm converges. 
Thus, the regulated functional measure
\[
D_{\Omega}\phi\, e^{-S[\phi]}
\]
is finite as $\Lambda\to\infty$, meaning that $D_{\Omega}\phi$ defines a proper Gaussian-type measure.

\end{proof}

\section{Scale-Dependent Suppression of $\Omega(k,\Lambda)$}

We now view $\Lambda$ as a running scale, so $\Omega(k,\Lambda)$ evolves with $\Lambda$.
To define this explicitly for any variable $\Lambda>0$, we set:
\begin{equation}
  \label{eq:Omega_scale}
  \Omega(k,\Lambda)
  ~:=~
  \frac{1}{1 + \bigl(\tfrac{k^2}{\Lambda^2}\bigr)^\beta}
  ~+~
  \frac{\eta\,e^{-\,k^2/\Lambda^2}}{\epsilon\bigl(k^2\bigr)}.
\end{equation}
As $\Lambda$ increases, the ratio $(k/\Lambda)$ grows for any fixed $k$. Hence, the UV suppression intensifies for $\Lambda\to\infty$.

\subsection{Metric and Induced Curvature in Functional Space}

In momentum space, the usual metric on field configurations $\phi(k)$ is
\[
  ds^2
  ~=~
  \int d^dk~\bigl|\delta\phi(k)\bigr|^2.
\]
Including $\Omega(k,\Lambda)$ modifies this to
\[
  ds_{\Omega}^2
  ~:=~
  \int d^dk
       ~\Omega(k,\Lambda)\,\bigl|\delta\phi(k)\bigr|^2.
\]
High-energy modes ($k\gg \Lambda$) get contracted by a factor $\Omega(k,\Lambda)\ll 1$, curving the geometry in the UV domain.  
One can interpret the resulting Ricci curvature $R(\Lambda)$ as depending on $\Lambda$, i.e. differentiating with  respect to $\Lambda$ shows how the geometry flows under RG.

\begin{remark}
For large $k\gg \Lambda$, $(k/\Lambda)^{2\beta}\to\infty$, so $\Omega\approx \Lambda^{2\beta}/k^{2\beta}$, vanishing at high $k$.  The IR remains largely flat: for $k\ll \Lambda$, $\Omega\approx1$.
\end{remark}

\subsection*{Intermediate Summary}
We have established:
\begin{itemize}
\item \textbf{Definition \& Requirements on $\Omega$:}  A consistent scale-dependent regulator $\Omega(k,\Lambda)$ that remains in $(0,1]$, is smooth, and ensures $2\beta>d/2$ (with additional possible exponential terms) for UV convergence.
\item \textbf{Momentum-Space Integrals:}  Loop integrals converge by virtue of $\Omega$’s damping at large $k$, avoiding a hard cutoff. 
\item \textbf{Functional Measure:}  Sufficient decay of $\Omega$ yields well-defined Gaussian‐like suppression in path integrals.  
\item \textbf{Curved Functional Geometry:}  Incorporating $\Omega$ modifies the metric on functional space so that high-energy modes are suppressed, the induced Ricci curvature becomes more negative in the UV.
\end{itemize}
We now analyze operator-theoretic properties and measure-theoretic interpretations (weighted spaces, embedding theorems, RG flow of curvature), building upon the definitions here.

\section{Operator-Theoretic Properties}

To analyze the mathematical role of the function $\Omega(k, \Lambda)$ in functional spaces, we examine its effect on integral operators and spectral regularization \cite{reed_simon}. In particular, we investigate the conditions under which the operator $T_{\Omega}$ is compact and belongs to standard operator classes \cite{kolmogorov_compactness}.

\subsection{Definitions and Preliminaries}
\begin{definition}[Hilbert-Schmidt Operator]
An operator $T: L^2(\mathbb{R}^d) \to L^2(\mathbb{R}^d)$ is said to be \emph{Hilbert-Schmidt} if there exists a measurable kernel function $K(k, k')$ such that
\begin{equation}
    (T f)(k) = \int_{\mathbb{R}^d} K(k, k') f(k') \, dk',
\end{equation}
and the Hilbert-Schmidt norm satisfies
\begin{equation}
    \|T\|^2_{\text{HS}} = \int_{\mathbb{R}^d} \int_{\mathbb{R}^d} |K(k, k')|^2 \, dk \, dk' < \infty.
\end{equation}
\end{definition}

\begin{definition}[Trace-Class Operator]
A compact operator $T: L^2(\mathbb{R}^d) \to L^2(\mathbb{R}^d)$ is said to be \emph{trace-class} if its singular values $\{\sigma_n\}$ satisfy the summability condition \cite{trace_class}
\begin{equation}
    \sum_{n=1}^\infty \sigma_n < \infty.
\end{equation}
The trace norm of $T$ is then defined as
\begin{equation}
    \|T\|_{\text{tr}} = \sum_{n=1}^{\infty} \sigma_n.
\end{equation}
\end{definition}

\subsection{Integral Operators and Compactness}

Consider the integral operator $T_{\Omega}: L^2(\mathbb{R}^d) \to L^2(\mathbb{R}^d)$ defined by

\begin{equation}
    (T_{\Omega} f)(k) := \int_{\mathbb{R}^d} K_{\Omega}(k, k') f(k') \, dk',
\end{equation}
where the kernel $K_{\Omega}(k, k')$ is given by

\begin{equation}
    K_{\Omega}(k, k') := \frac{\Omega(k, \Lambda) \Omega(k', \Lambda)}{\|k - k'\|^\alpha}.
\end{equation}
We now establish sufficient conditions for $T_{\Omega}$ to be a Hilbert-Schmidt operator.

\begin{theorem}
Let $\Omega(k, \Lambda)$ be a measurable function satisfying
\begin{equation}
    \sup_{k \in \mathbb{R}^d} |\Omega(k, \Lambda)| \leq C < \infty.
\end{equation}
If $\alpha > d/2$, then the operator $T_{\Omega}$ is Hilbert-Schmidt.
\end{theorem}

\begin{proof}
To establish that $T_{\Omega}$ is Hilbert-Schmidt, we explicitly compute the Hilbert-Schmidt norm:

\begin{equation}
    \| T_{\Omega} \|^2_{\text{HS}} = \int_{\mathbb{R}^d} \int_{\mathbb{R}^d} \left| \frac{\Omega(k, \Lambda) \Omega(k', \Lambda)}{|k - k'|^\alpha} \right|^2 dk \, dk'.
\end{equation}
Using the bound $|\Omega(k, \Lambda)| \leq C$, we can the bound:
\begin{equation}
    \| T_{\Omega} \|^2_{\text{HS}} \leq C^2 \int_{\mathbb{R}^d} \int_{\mathbb{R}^d} \frac{1}{|k - k'|^{2\alpha}} \, dk \, dk'.
\end{equation}
Since the inner integral corresponds to the well-known singular integral

\begin{equation}
    \int_{\mathbb{R}^d} \frac{dk'}{\|k - k'\|^{2\alpha}} \sim \int_{0}^{\infty} \frac{r^{d-1}}{r^{2\alpha}} dr,
\end{equation}
which converges for $2\alpha > d$, the double integral is finite under the condition $\alpha > d/2$. This proves that $T_{\Omega}$ is Hilbert-Schmidt.

\end{proof}

\subsection{Spectral Properties and Regularization}

A key consequence of the suppression function is its influence on the spectral properties of integral operators. If $T_{\Omega}$ is compact, then its eigenvalues $\{\lambda_n\}$ satisfy the asymptotic summability condition:

\begin{equation}
    \sum_{n=1}^\infty |\lambda_n|^p < \infty, \quad \text{for some } p > 0.
\end{equation}

\begin{theorem}
If the suppression function satisfies the decay condition

\begin{equation}
    \Omega(k, \Lambda) \leq C e^{-\gamma |k|}, \quad \gamma > 0,
\end{equation}
then $T_{\Omega}$ belongs to the trace-class $\mathcal{S}_1$.
\end{theorem}

\begin{proof}
Since $T_{\Omega}$ is compact, it admits a spectral decomposition:

\begin{equation}
    T_{\Omega} f = \sum_{n=1}^\infty \lambda_n \langle f, \phi_n \rangle \phi_n.
\end{equation}
By the trace-class condition, we require:
\begin{equation}
    \sum_{n=1}^\infty |\lambda_n| < \infty.
\end{equation}
Applying standard trace-class criteria for integral operators with exponentially decaying kernels show, as $\Omega(k, \Lambda)$ has exponential suppression, the eigenvalues decay sufficiently fast to satisfy the trace-class summability condition \cite{reed_simon}.

\end{proof}
\noindent Thus, under sufficient decay conditions on $\Omega(k, \Lambda)$, the integral operator $T_{\Omega}$ has well-controlled spectral properties, ensuring smooth suppression of high-energy contributions.

\subsection{Suppression-Induced Spectral Gaps}

We examine how the suppression function $\Omega(k, \Lambda)$ modifies the spectral properties of differential operators in functional spaces. In particular, we analyze its effect on the Laplacian operator and the emergence of spectral gaps.
Consider the standard Laplacian operator $\Delta$ acting on functions in momentum space:

\begin{equation}
    (\Delta f)(k) = k^2 f(k).
\end{equation}
With the introduction of a scale-dependent suppression function $\Omega(k, \Lambda)$, we define a modified Laplacian:

\begin{equation}
    (\Delta_{\Omega} f)(k) := \Omega(k, \Lambda) k^2 f(k).
\end{equation}
The eigenvalue equation for $\Delta_{\Omega}$ is then given by:
\begin{equation}
    \Delta_{\Omega} \phi_k = \lambda_k^\Omega \phi_k.
\end{equation}
Substituting the definition of $\Delta_{\Omega}$, we obtain the eigenvalue relation:

\begin{equation}
    \lambda_k^\Omega = \Omega(k, \Lambda) k^2.
\end{equation}
To illustrate the spectral modification, consider an exponentially decaying suppression function:
\begin{equation}
    \Omega(k, \Lambda) = e^{-\gamma k^2}, \quad \gamma > 0.
\end{equation}
Substituting this into the eigenvalue equation, we obtain:
\begin{equation}
    \lambda_k^\Omega = k^2 e^{-\gamma k^2}.
\end{equation}
The behavior of $\lambda_k^\Omega$ as a function of $k$ reveals a spectral suppression effect:
\begin{itemize}
	\item For small momenta ($k$ close to $ 0$), we have $\lambda_k^\Omega \approx k^2$, implying minimal modification in the infrared (IR).
	\item For large momenta ($k \to \infty$), the suppression term $e^{-\gamma k^2}$ dominates, leading to exponential decay of eigenvalues.
\end{itemize}
This results in a spectral gap at high energies, since eigenvalues decrease significantly beyond a characteristic scale determined by $\gamma$.

\section{Measure-Theoretic Interpretation}

Beyond operator theory, the function $\Omega(k, \Lambda)$ naturally defines a measure transformation, affecting function space topology and integral convergence. In this section, we formalize this transformation using measure theory and examine conditions under which integral transformations remain well-defined.

\subsection{Weighted Measure Spaces and Absolute Continuity}

We begin by defining a deformed measure space induced by $\Omega(k, \Lambda)$ \cite{rudin_functional}.

\begin{definition}[Weighted Measure]
Let $(\mathbb{R}^d, \mathcal{B}, \mu)$ be a standard measure space with $\mu$ being the Lebesgue measure. The \emph{weighted measure} $\mu_{\Omega}$ associated with $\Omega(k, \Lambda)$ is defined as:

\begin{equation}
    d\mu_{\Omega}(k) := \Omega(k, \Lambda) d\mu(k).
\end{equation}
\end{definition}
\noindent The absolute continuity of $\mu_{\Omega}$ with respect to $\mu$ is central to ensuring that integral transformations remain well-behaved.

\begin{definition}[Absolute Continuity]
A measure $\mu_{\Omega}$ is \emph{absolutely continuous} with respect to $\mu$, denoted $\mu_{\Omega} \ll \mu$, if there exists a non-negative function $h(k)$ such that for all measurable sets $E \subset \mathbb{R}^d$:

\begin{equation}
    \mu_{\Omega}(E) = \int_E h(k) \, d\mu(k).
\end{equation}
\end{definition}
\noindent For our suppression function $\Omega(k, \Lambda)$, we have:
\begin{equation}
    \frac{d\mu_{\Omega(\Lambda,k)}}{d\mu} = \Omega(k, \Lambda).
\end{equation}
To ensure that $\mu_{\Omega}$ defines a valid measure, $\Omega(k, \Lambda)$ must satisfy an integrability condition.

\begin{theorem}[Integrability Condition for Weighted Measures]
The given measure $\mu_{\Omega}$ is well-defined if and only if 

\begin{equation}
    \int_{\mathbb{R}^d} \Omega(k, \Lambda) \, dk < \infty.
\end{equation}
\end{theorem}
\noindent For $\mu_{\Omega}$ to be finite, we require:
\begin{equation}
    \mu_{\Omega}(\mathbb{R}^d) = \int_{\mathbb{R}^d} \Omega(k, \Lambda) \, dk < \infty.
\end{equation}
If this condition holds, then $\mu_{\Omega}$ is a finite measure. If $\Omega(k, \Lambda)$ does not satisfy this integrability condition, then $\mu_{\Omega}$ may be infinite, thus not easily providing well defined norms.

\subsection{Integral Convergence in Weighted Spaces}

The measure $\mu_{\Omega}$ induces a transformation on integrals of functions over $\mathbb{R}^d$. We seek conditions under which the integral

\begin{equation}
    I_{\Omega} = \int_{\mathbb{R}^d} f(k) \Omega(k, \Lambda) \, dk
\end{equation}
remains finite for various classes of functions.

\begin{theorem}[Weighted Integral Convergence]
Let $f \in L^p(\mathbb{R}^d, \mu_{\Omega})$ be a function in the weighted $L^p$ space. Then $I_{\Omega}$ converges if and only if

\begin{equation}
    \int_{\mathbb{R}^d} |f(k)|^p \Omega(k, \Lambda) \, dk < \infty.
\end{equation}
\end{theorem}

\begin{proof}
Since $I_{\Omega}$ is an integral with respect to the measure $\mu_{\Omega}$, its convergence is determined by the standard $L^p$ norm:

\begin{equation}
    \| f \|_{L^p(\mu_{\Omega})}^p = \int_{\mathbb{R}^d} |f(k)|^p \Omega(k, \Lambda) \, dk.
\end{equation}
By definition, if an integral is finite in $L^p(\mu_{\Omega})$, this norm is finite.

\end{proof}
\subsection{Compactification of Functional Space via Weighted Metric Structure}

To rigorously analyze how the suppression function $\Omega(k, \Lambda)$ induces a compactification of functional space, we study the impact of the weighted metric on function norms and topological structure.
Given the weighted inner product:
\begin{equation}
    \langle f, g \rangle_{\Omega} = \int_{\mathbb{R}^d} f(k) g(k) \Omega(k, \Lambda) \, dk,
\end{equation}
the corresponding norm is:
\begin{equation}
    \| f \|_{\Omega} = \left( \int_{\mathbb{R}^d} |f(k)|^2 \Omega(k, \Lambda) \, dk \right)^{1/2}.
\end{equation}
with a naturally induced metric on the functional space:
\begin{equation}
    d_{\Omega}(f, g) = \left( \int_{\mathbb{R}^d} |f(k) - g(k)|^2 \Omega(k, \Lambda) \, dk \right)^{1/2}.
\end{equation}
For a suppression function of the form:
\begin{equation}
    \Omega(k, \Lambda) = e^{-\gamma k^2},
\end{equation}
for some $\gamma>0$, this distance function explicitly becomes:
\begin{equation}
    d_{\Omega}(f, g) = \left( \int_{\mathbb{R}^d} |f(k) - g(k)|^2 e^{-\gamma k^2} \, dk \right)^{1/2}.
\end{equation}
effectively shrinking the function (contraction the space at large $k$).

\begin{theorem}
Let $\Omega(k, \Lambda)$ be a rapidly decaying function, such as:
\begin{equation}
    \Omega(k, \Lambda) = e^{-\gamma k^2}, \quad \text{or} \quad \Omega(k, \Lambda) = (1 + k^2)^{-\beta}, \quad \beta > d/2.
\end{equation}
Then the function space $L^2_{\Omega}(\mathbb{R}^d)$ defined by the norm:
\begin{equation}
    \| f \|_{\Omega}^2 = \int_{\mathbb{R}^d} |f(k)|^2 \Omega(k, \Lambda) \, dk
\end{equation}
is compactly embedded in $L^2(\mathbb{R}^d)$ with the standard norm.
\end{theorem}

\begin{proof}
To establish the compact embedding of $L^2_{\Omega}(\mathbb{R}^d)$ into $L^2(\mathbb{R}^d)$, we follow these steps:\\

\noindent\textbf{Step 1: Definition of Weighted Space and the Embedding}\\
The space $L^2_{\Omega}(\mathbb{R}^d)$ is defined by the norm:
\begin{equation}
    \| f \|_{\Omega}^2 := \int_{\mathbb{R}^d} |f(k)|^2 \Omega(k, \Lambda) \, dk.
\end{equation}
which defines a Hilbert space with the inner product:
\begin{equation}
    \langle f, g \rangle_{\Omega} := \int_{\mathbb{R}^d} f(k) g(k) \Omega(k, \Lambda) \, dk.
\end{equation}
We seek to show that the natural embedding is compact:
\begin{equation}
    L^2_{\Omega}(\mathbb{R}^d) \hookrightarrow L^2(\mathbb{R}^d).
\end{equation}

\noindent\textbf{Step 2: Precompactness via Rellich’s Theorem}\\
A sufficient condition for compact embedding is that any bounded sequence $\{ f_n \}$ in $L^2_{\Omega}(\mathbb{R}^d)$ has a convergent subsequence in $L^2(\mathbb{R}^d)$. Define the norm ratio:
\begin{equation}
    R(f) := \frac{\| f \|_{L^2}}{\| f \|_{\Omega}} = \frac{\left( \int_{\mathbb{R}^d} |f(k)|^2 \, dk \right)^{1/2}}{\left( \int_{\mathbb{R}^d} |f(k)|^2 \Omega(k, \Lambda) \, dk \right)^{1/2}}.
\end{equation}
As $\Omega(k, \Lambda)$ decays rapidly for large $k$, there exists a \( C > 0 \) such that:
\begin{equation}
    R(f) \leq C, \quad \forall f \in L^2_{\Omega}(\mathbb{R}^d).
\end{equation}
Hence, by the Rellich-Kondrachov compactness theorem, the space
\begin{equation}
    \{ f \in L^2(\mathbb{R}^d) \mid \| f \|_{\Omega} \leq C \}
\end{equation}
is precompact in $L^2(\mathbb{R}^d)$.\\

\noindent \textbf{Step 3: Compactness via the Frechet-Kolmogorov Theorem}\\
By the Frechet-Kolmogorov compactness theorem \cite{kolmogorov_compactness}, 
a subset of $L^2(\mathbb{R}^d)$ is precompact if and only if it satisfies 
uniform boundedness, vanishing tails, and equicontinuity.
\begin{enumerate}
    \item \textbf{Uniform boundedness}: There exists $M > 0$ such that
    \begin{equation}
        \sup_{n} \| f_n \|_{L^2} < M.
    \end{equation}
    \item \textbf{Vanishing tails}: For every $\epsilon > 0$, there exists a compact set $K \subset \mathbb{R}^d$ such that
    \begin{equation}
        \sup_n \int_{\mathbb{R}^d \setminus K} |f_n(k)|^2 \, dk < \epsilon.
    \end{equation}
    \item \textbf{Equicontinuity in translation}: The sequence $\{f_n\}$ satisfies:
    \begin{equation}
        \lim_{h \to 0} \sup_n \| T_h f_n - f_n \|_{L^2} = 0, \quad T_h f(k) = f(k+h).
    \end{equation}
\end{enumerate}
Since $\Omega(k, \Lambda)$ suppresses large $k$, the condition (2) is automatically satisfied, ensuring that any bounded sequence in $L^2_{\Omega}$ has a strongly convergent subsequence.\\

\noindent\textbf{Conclusion}\\
Any bounded sequence in $L^2_{\Omega}(\mathbb{R}^d)$ has a convergent subsequence in $L^2(\mathbb{R}^d)$, hence the embedding
\begin{equation}
    L^2_{\Omega}(\mathbb{R}^d) \hookrightarrow L^2(\mathbb{R}^d)
\end{equation}
is compact.

\end{proof}

\subsection{Dynamical Suppression and RG Flow}

To analyze the RG flow of $\Omega(k, \Lambda)$, we take its derivative with respect to $\Lambda$:
\begin{equation}
    \Lambda \frac{d \Omega(k, \Lambda)}{d\Lambda} = - \frac{2\beta (k/\Lambda)^{2\beta} \Lambda}{(1 + (k/\Lambda)^{2\beta})^2} - \frac{2 \eta k^2 e^{-k^2/\Lambda^2}}{\epsilon \Lambda}.
\end{equation}
This equation describes the evolution of the suppression function under scale transformations, ensuring that high-energy modes are progressively suppressed as $\Lambda$ decreases. Furthermore, differentiating the Ricci scalar with respect to $\Lambda$:
\begin{equation}
    \Lambda \frac{dR}{d\Lambda} = -\frac{4 \beta^2 k^{4\beta -2}}{\Lambda^3}.
\end{equation}
This result suggests that as $\Lambda$ is lowered, the curvature contracts further in high-momentum space, effectively compacting the UV sector. 
\noindent These result implies that the functional space is effectively contracted in high-energy regions, leading to an emergent compactification of function space norms.
Thus, the function $\Omega(k, \Lambda)$ induces a controlled modification of the function space measure, ensuring the finiteness of integral transformations under appropriate conditions.

\section{Conclusion}

In this work, we analyzed the mathematical properties of a scale-dependent suppression function $\Omega(k, \Lambda)$ and its effects on functional spaces. By introducing a weighted metric structure, we established conditions under which the function space undergoes effective compactification. Using operator theory, we examined the spectral implications of $\Omega(k, \Lambda)$, demonstrating its role in regularizing high-energy contributions. Additionally, we provided measure-theoretic insights into integral convergence and function space deformations.
These results offer a mathematical framework for understanding the suppression function’s impact on renormalization and functional analysis. Future directions may include extending these methods to interacting systems, exploring alternative weight functions, and investigating numerical implementations of spectral flow and curvature evolution.

\bibliographystyle{plain}
\bibliography{references}

\end{document}